\pdfoutput=1
\documentclass{article}
\usepackage{psfrag,epsf}
\usepackage{url} 

\usepackage{graphicx,amssymb,amsthm,amsmath,amsfonts}
\usepackage{epstopdf}
\usepackage{mathtools}
\usepackage{xcolor}
\usepackage{lipsum}
\usepackage{algorithmic,algorithm2e}
\usepackage{subcaption}

\usepackage{enumitem}
\setlist[enumerate]{leftmargin=.5in}
\setlist[itemize]{leftmargin=.5in}

\newtheorem{theorem}{Theorem}[section]

\newtheorem{lemma}[theorem]{Lemma}

\newtheorem{proposition}[theorem]{Proposition}
\newtheorem{definition}[theorem]{Definition}

\newtheorem{assumption}[theorem]{Assumption}

\newcommand{\blind}{1}

\begin{document}

\def\spacingset#1{\renewcommand{\baselinestretch}%
{#1}\small\normalsize} \spacingset{1}



\if1\blind
{
  \title{\bf Optimal Shape Control via $L_\infty$ Loss for Composite Fuselage Assembly}
  \author{Juan Du$^1$, Shanshan Cao$^1$, Jeffrey H. Hunt$^2$, Xiaoming Huo$^1$ \thanks{
    The authors gratefully acknowledge \textit{the Strategic University Partnership between the Boeing Company and Georgia Institute of Technology}. The corresponding author is Xiaoming Huo.}\hspace{.2cm} and Jianjun Shi$^1$\\
    $^1$H. Milton Stewart School of Industrial and Systems Engineering,\\ Georgia Institute of Technology\\
    $^2$The Boeing Company}
  \maketitle
} \fi

\if0\blind
{
  \bigskip
  \bigskip
  \bigskip
  \begin{center}
    {\LARGE\bf Optimal Shape Control via $L_\infty$ Loss for Composite Fuselage Assembly}
\end{center}
  \medskip
} \fi

\bigskip
\begin{abstract}
Natural dimensional variabilities of incoming fuselages affect the assembly speed and quality of fuselage joins in composite fuselage assembly processes. 
Shape control is critical to ensure the quality of composite fuselage assembly. 
In current practice, the structures are adjusted to the design shape in terms of the $\ell_2$ loss for further assembly without considering the existing dimensional gap between two structures. 
Such practice has two limitations: (1) the design shape may not be the optimal shape in terms of a pair of incoming fuselages with different incoming dimensions; (2) the maximum gap is the key concern during the fuselage assembly process, so the $\ell_\infty$ loss of gap after control oughts to be considered. This paper proposes an optimal shape control methodology via the $\ell_\infty$ loss for composite fuselage assembly process by considering the existing dimensional gap between the incoming pair of fuselages. On the other hand, due to the limitation on the number of available actuators in practice, we face an important problem of finding the best locations for the actuators among many potential locations, which makes the problem a sparse estimation problem.
We are the first to solve the optimal shape control in fuselage assembly process using the $\ell_\infty$ model under the framework of sparse estimation, where we use the $\ell_1$ penalty to control the sparsity of the resulting estimator.
From statistical point of view, this can be formulated as the $\ell_\infty$ loss based linear regression, and under some standard assumptions, such as the restricted eigenvalue (RE) conditions, and the light tailed noise, the non-asymptotic estimation error of the $\ell_1$ regularized  $\ell_\infty$  linear model is derived to be the order of $O(\sigma\sqrt{\frac{S\log p}{n}})$, which meets the upper-bound in the existing literature.
Compared to the current practice, the case study shows that our proposed method significantly reduces the maximum gap between two fuselages after shape adjustments.


\end{abstract}

\noindent%
\vfill

\newpage
\spacingset{1.45} 

\section{Introduction}
Recently, composite materials are widely used in large space structures due to its superior properties such as high strength-to-weight ratio. For example, an airplane of Boeing 787 comprises more than 50$\%$ composite parts by weight, and the fuselage is one key composite part of an aircraft \cite{chawla2012composite}. In practice, there are natural dimensional variations in the fabrication of a fuselage due to different manufacturing batches or different suppliers \cite{gates2007boeing}. When two fuselages assemble together, there is a gap, as shown in Figure \ref{fig:Illustrations}. Such gap will infuence the assembly quality and productivity of fuselage assembly.
Thus, dimension variations influence the speed and the quality of composite fuselage assembly. As a result, the interface between two fuselages requires shape adjustments before the composite fuselage assembly process. 

\begin{figure}[htbp]
    \centering
        \includegraphics[width=0.8\textwidth]{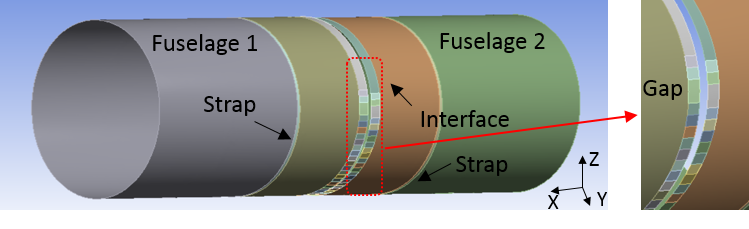}
    \caption{Illustrations of composite fuselage assembly.}
    \label{fig:Illustrations}
\end{figure}

In practice, actuators are used for shape adjustments of the interface between two composite fuselages, as shown in Figure \ref{fig:actuator}. More details about using actuators for shape adjustments of composite fuselages can be found in \cite{wen2018feasibility}, \cite{yue2018surrogate}. In this paper, we focus on the shape adjustments of interface that are close to the edge plane of the fuselage. In the state of the art, the target shape after control is the design shape, as shown in Figure \ref{fig:interfaces}. The dashed line indicates the design shape and the solid line shows the shape of an incoming fuselage. The arrows represent the positions of the actuators used for shape control of the interface. The current shape control strategy regards the target shape after adjustments as the design shape in terms of the $\ell_2$ loss, which has two key limitations. (1) It is non-optimal. For a given pair of incoming fuselages with specific shapes, adjustment to the design shape is not optimal. In other words, there is no guarantee that the optimal shape control can be realized by shape adjustments to the design shape for two different incoming fuselages. Here, optimality means to achieve the minimum of maximum gap on the interface between two fuselages after shape control. (2) For fuselage assembly, the maximum gap between two fuselages is the key quality indicator during the fuselage assembly process. Thus, the  $\ell_\infty$ loss should be considered.

\begin{figure}[htbp]
    \centering
        \includegraphics[width=0.5\textwidth]{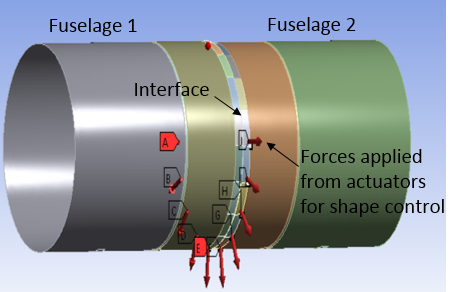}
    \caption{An illustration of shape adjustments by using actuators.}
    \label{fig:actuator}
\end{figure}

\begin{figure}[htbp]
    \centering
        \includegraphics[width=0.8\textwidth]{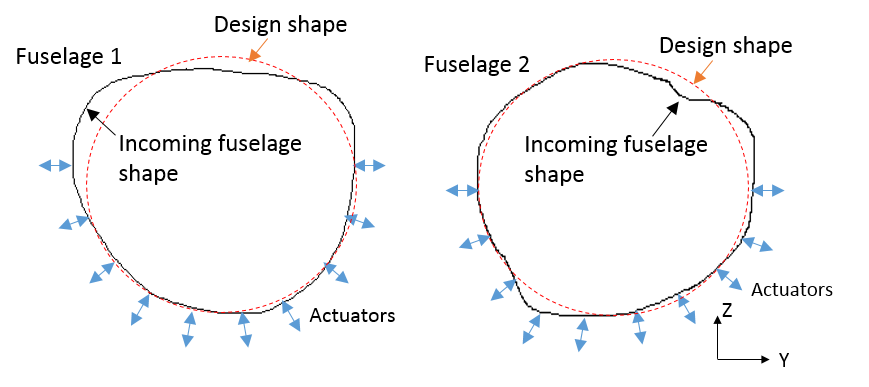}
    \caption{Schematics of interfaces between two fuselages.} 
    \label{fig:interfaces}
\end{figure}

In the literature, multiple efforts are made to achieve shape control for structures. To focus on the research work related to this paper, we mainly introduce the literature for composite fuselage shape control. 

For the shape control of composite fuselages, Wen et al. \cite{wen2018feasibility} first develop a new shape control system based on the Finite Element Analysis (FEA) to improve the dimensional quality and productivity. Also, they show the feasibility of shape control for composite fuselages by using the proposed FEA platform. Based on this FEA platform, Yue et al. \cite{yue2018surrogate} propose a surrogate model-based control strategy by considering the uncertainties for composite fuselage assembly. By minimizing the $\ell_2$ loss of dimensional errors in relative to the design shape, they calculate the optimal forces applied from actuators for single fuselage shape control. Based on the surrogate model, Du et al. \cite{du2019fuselage} propose an optimal actuator placement strategy for the shape adjustment of composite fuselages; their strategy reduces the forces applied from actuators and also dimensional deviations after shape control. However, all these works only consider shape adjustments of a single fuselage, and aim to adjust the incoming fuselages to the design shapes in terms of the  $\ell_2$ loss, which meets the requirement of a single shape adjustment. However, when two fuselages assemble together, the maximum gap, i.e.,  the $\ell_\infty$ norm of the gap, is one of the most important quality indicator in a fuselage assembly process. In terms of the maximum gap reduction for composite fuselage assembly, adjusting incoming fuselages to design shapes cannot guarantee the optimality given a specific pair of fuselages with different dimensional variations. Hence, a better shape control strategy is needed in the composite fuselage assembly process. 

To fill the aforementioned gap, this paper proposes an optimal shape control strategy for composite fuselage assembly. We consider minimizing the maximum of dimensional gap between a pair of incoming fuselages.  Instead of adjusting each fuselage into the design shape, we  consider the initial gap between the pair of fuselages and optimize the adjustment into an intermediate shape. 

As shown in \cite{du2019fuselage}, one other direct result is that the optimal  placement of actuators can also lead to the improvement of the performance in shape control.
In the current paper, as a part of the optimal shape control strategy for the composite fuselage assembly, we also consider the optimal  placement of actuators for two fuselage assembly. 
Given the dimensions of the pair of incoming fuselages, a sparse learning model is proposed to link actuator forces with the weighted maximum gap deviation of a pair of fuselages. From a statistical point of view, the proposed sparse learning model is the $\ell_1$ sparsity penalized $\ell_\infty$ loss linear regression. We also contribute to analyze the properties of the resulting estimator under a light-tailed sub-Gaussian errors assumption. Specifically, we provide the non-asymptotic upper-bound of the estimation error and prediction error in the sparse estimation problem. An alternating direction method of multipliers (ADMM)-based algorithm is derived to solve the optimization problem. The nonzero components of the force vector correspond to the optimal actuator locations. Often in practice, the last step is to refit the model with only the selected locations and get a more accurate estimation on the unknown forces. Consequently, the optimal forces applied to the pair of fuselages can be obtained by minimizing the weighted maximum gap deviation. Finally, case studies of composite fuselage assembly process are used to verify the effectiveness of our proposed method.

\subsection{Notations}
Throughout this work, we will need the following notations.
We will denote vectors/matrix as $Y$, $X$, $\beta$, $\epsilon$, etc, and all the vectors are the column vectors by default. 
For a matrix $A \in \mathbb{R}^{m \times n}$, $A_{ij}$ denotes the element of $A$ at the $(i,j)$th location, and $A_{i}$ denotes the $i$th column  of $A$.
Similarly, for a vector $\beta \in \mathbb{R}^{m}$, $\beta_{i}$ denotes the $i$th element of $\beta$.
For a subset $S \subset \{1, \cdots, m\}$, $A_{S}$ denotes sub-matrix of $A$ containing the columns whose indices are in $S$, and $\beta_{S}$ denotes a sub-vector of $\beta$ that contains the elements whose indices are in $S$.
We will use   $\|\cdot\|_p$ to denote the $\ell_p$ norm. 
Specifically, for a vector $x \in \mathbb{R}^n$, the $\ell_p$ norm of $x$ is defined as: $\|x\|_p := (\sum_{i = 1}^{n}|x_i|^p)^{\frac{1}{p}}$.
According to this definition, the $\ell_\infty$ norm of $x$ is simply $\|x\|_\infty = \max_{i = 1}^{n}|x_i|$.
In the case of the Euclidean norm,  which is also known as the $\ell_2$ norm, we will simply omit the subscript $p$ $(p=2)$, i.e., $\|\cdot\|$.
We will denote the entrywise max norm of a matrix $A \in \mathbb{R}^{m \times n}$ as $\|A\|_\infty = \max_{1\leq i \leq m, 1\leq j \leq n }|A_{ij}|$.
Throughout this work, we use the $\mbox{sign}(\cdot)$ function to indicate the sign of a vector or scalar.
For $x \in  \mathbb{R}$, $\mbox{sign}(x) = 1$ if $x>0$ and  $\mbox{sign}(x) = -1$ otherwise.
For a vector $\beta \in \mathbb{R}^{m}$, $\mbox{sign}(\beta)$  is a vector with the $i$th element equal to $\mbox{sign}(\beta_i)$.
For a vector $\beta \in \mathbb{R}^{m}$, the support of $\beta$ is defined as $\mbox{supp}(\beta) = \{i: \beta_i \neq 0\}$.
And $|\mbox{supp}(\beta)|$ denote the cardinality of the set $\mbox{supp}(\beta)$.

\subsection{Outline}

In this work, we first provide a detailed description of the physical model that is of concern, including our justifications of using a linear model with the $\ell_\infty$ loss function in Section \ref{sec:fuselagemodel}.
Then, we investigate the theoretical properties of the proposed method, and provide main theoretical results in Section \ref{sec:model}. 
In Section \ref{sec:alg}, we propose an efficient algorithm to solve the proposed $\ell_\infty$ optimization problem. Our method is based on ADMM.
A case study on fuselage assembly process using the FEA generated data  is studied in Section \ref{sec:numeric}, which validates the effectiveness of our proposed approach in this type of engineering problems.
Finally, we conclude our work in  Section \ref{sec:conclusion}.
All the proofs of our main theorems are provided in the Appendix.

\section{Physical model}\label{sec:fuselagemodel}

In this section, before going to the statistical analysis, we provide necessary background of our physical model. 
For the composite fuselage shape control, only elastic deformation is allowed during the shape adjustment. In this paper, we assume linear mechanical behavior of fuselage deformation that corresponds to the actuator forces. Consequently, the adjusted shape deviations can be formulated as 
\begin{equation}
\delta_i = \psi_i + U_iF_i, \qquad i=1,2.
\end{equation}
where $i$ indicates the $i$th fuselage, $\delta_i \in \mathbb{R}^{2n}_{i}$ is the error in Y and Z directions after shape control and $n$ denotes the number of measurement points for each fuselage; $\psi_i$ and  $U_i \in \mathbb{R}^{2n \times m_i}$ represent dimension deviations and displacement matrix of incoming fuselage $i$, respectively, and $m_i$ denotes the number of feasible positions (e.g., candidate positions where an actuator may be placed) for actuators in the $i$th fuselage edge plane; $F_i$ is the applied force during shape control  of the fuselage $i$.  The physical interpretation of the  displacement matrix $U_i$ is that $U_i$ is the deformation of all the measurement points correspond to the unit force on the structure. 
Without loss of generality, we assume that the locations of measurement points are the same. Therefore, the dimensional gap between two fuselages after shape adjustment can be written as
\begin{equation}
 		\Delta=\delta_2-\delta_1= \psi_2+U_2F_2-(\psi_1+U_1F_1).
\end{equation}
Similar to \cite{du2019fuselage}, $U_i, i=1, 2,$ can be obtained from the surrogate model. Notably, the registration between two fuselages is needed if the number of measurement points for the pair of fuselages is not the same, which can be easily solved through a method that is proposed in \cite{chui2000new}. 

For the composite fuselage assembly process, the main concern before assembly is the maximum gap and its location along the interface between the two fuselages. In this paper, our objective is to minimize the weighted maximum gap between the pair of fuselages in both Y and Z directions, i.e., $\Delta_{max} $, which is defined as 
\begin{equation}
 		\Delta_{max}=\|B \Delta\|_\infty, 
\end{equation}
where $B \in \mathbb{R}^{n \times n}$ is a diagonal matrix of weights, which represent the importance of the gaps in different measurement points. Such weight matrix is determined based on some engineering knowledge. For example, if we emphasize the dimensional gap on the upper fuselage, we can add more weight in the part of the B matrix that correspond to the upper fuselage.

Usually in practice, we only have a limited number of actuators available. Thus, among many potential choices of positions of actuators, we want to find the most effective ones, which may correspond to a sparse solution.  Recall that the optimal actuator locations are the nonzero locations in the solution. Let $m_1$ and $m_2$ denote  the numbers of all feasible locations for actuators along the pair of fuselages, while only a total number of  $M$ actuators are available for the shape control process in practice. 
In order to encourage the sparsity in the resulting solution, we add a sparsity-induced penalty, namely, the $\ell_1$ penalty, on the unknown forces at different locations:
\begin{equation}
\label{eq:objforce_l1}
\begin{split}
\min_{F_1,F_2} &J :=\|B\Delta\|_\infty +  \lambda \|F_1\|_1 +\lambda \|F_2\|_1,
\end{split}
\end{equation}
where  $\lambda$ is the tuning parameter, which controls the sparsity of the resulting estimation.



 From the solution of Problem (\ref{eq:objforce_l1}), the optimal actuator placement is obtained from the support of $F_1$ and $F_2$ for the shape adjustments of the pair of fuselages. 
In order to reduce the bias on the estimation of the forces,which are denoted by $F_1$ and $F_2$, we refit the model using the optimal positions and the optimal forces can be obtained by solving the following problem:  
\begin{equation}
\label{eq:objforce_step2}
\begin{split}
\min_{F_1,F_2} &J :=\|B(\psi_2 -  \psi_1 + (U_2)_{S_2}(F_2)_{S_2} - (U_1)_{S_1}(F_1)_{S_1}) \|_\infty \\
\end{split}
\end{equation}
where $S_1=\mbox{supp}(F_1)$, $S_2=\mbox{supp}(F_2)$.

 In order to make the above physical model more evident to the statisticans without engineering background, we rewrite the physical model as follows:
\begin{equation}
\label{eq:formulation1}
\min_\beta \|Y - X\beta\|_\infty+ \lambda \|\beta\|_1,
\end{equation}
where $Y=B(\psi_2 -  \psi_1 )$, $X=[BU_1, -BU_2]$, and $\beta=[F_1', F_2']'$. More details of rewriting the physical model as a statistical problem will be provided in Section \ref{sec:model}.

From the statistical perspective, the above formulations can be viewed as the linear regression problem with  $\ell_\infty$ loss and  $\ell_1$ regularization, and then refitting the model using only the nonzero locations, which helps reduce the bias on the estimation of the forces, i.e., $F_1$ and $F_2$. In this way, we can have a better estimation in practice.

In the statistical literature, the parameter estimation based on the $\ell_\infty$ loss has been studied since 1980, such as applications in the physical and environmental sciences \cite{james1983fitting,zolghadri2004minimax,qi2015theoretical}, signal processing and systems engineering \cite{milanese1982estimation, alecu2006wavelet, Castillo2009combined}, and so on.
However, the theoretical guarantee is very limited, especially in the sparse estimation framework. 

There are two main directions studying regression analysis in the statistical literature. The first one focuses on the asymptotic distribution in linear regression in a well-posed setting, where we have many more observations than the number of feature variables, and there is no assumption on the sparsity of the unknown signals.
Reference \cite{knightasymptotic} studies the problem of $\ell_\infty$ (or Chebyshev) estimator, which minimizes the maximum absolute residual under the condition where the noise distribution is known to have bounded support or light tails.
They derived the asymptotic distribution of the estimator in the low dimension setting.

The second one focuses on problems in the ill-posed settings, or settings where the unknown signal is sparse. In other words, there are far more feature variables than the number of observations or there are a lot of 
zero components in the coefficients.
The Dantzig selector proposed in \cite{candes2007dantzig} solves the following problem:
\begin{eqnarray}
\label{eq:dantzig}
\begin{split}
\min _{\beta \in \mathbb{R}^p} \quad &\|\beta\|_1,\\
s.t.\quad &\|X^T(Y - X\beta)\|_\infty \leq \lambda \sigma ,
\end{split}
\end{eqnarray}
where the constrained optimization problem seeks to minimize the $\ell_1$ sparsity objective function within the feasible region, where $\|X^T(Y - X\beta)\|_\infty \leq \lambda \sigma$. It can be easily seen that the term $X^T(Y - X\beta)$ is simply the first order derivative of the least square loss. 
Thus, the constraint $\|X^T(Y - X\beta)\|_\infty \leq \lambda \sigma$ will ensure that the scale of the first order derivative of the least square loss is very small.
Essentially, Dantzig selector solves a similar problem as the lasso, i.e., minimizing the least square loss \cite{james2009dasso}.

However, theories to the previous models do not apply to our problem, where we have the incentive to minimize $\|Y - X\beta\|_\infty$ loss. In Section \ref{sec:model}, we rewrite the above sparse estimation problem using the $\ell_\infty$ loss in the statistical language in details, and provide  theoretical results on the non-asymptotic upper-bounds of the estimation error and prediction error under the linear model setting with light tailed sub-Gaussian errors.

\section{Theoretical Results} \label{sec:model}
In this section, we rewrite the above physics model in statistical language to make the problem clearer with reader with no engineering background.
Specifically, we give the statistical formulation in Subsection \ref{subsec:model}.
Then we present our main theoretical results in Subsection \ref{sec:main}.
\subsection{Model in statistical language}\label{subsec:model}
We will consider the linear model throughout this work:
\begin{equation}
\label{eq:dataGeneration}
Y = X^T\beta + \epsilon.
\end{equation}
For all  observations, it is assumed that all elements in the vector $\epsilon$, i.e., $\epsilon_i$, for $i = 1, \cdots, n$, are independent sub-Gaussian random variables with common variance $\sigma^2$ and mean $0$.
As in our motivated example, we want to estimate the forces, which will minimize the $\ell_\infty$ norm of the gap between fuselages along the interface. Thus, we propose the following formulation of the estimation problem:
\begin{equation}
\label{eq:formulation}
\min_\beta \|Y - X\beta\|_\infty,
\end{equation} 
where $Y - X\beta$ is the estimated gap between two fuselages after control. 
In our motivated application, we believe the underlying forces are sparse to find the optimal locations of the actuators, i.e., the $\beta$ vector is sparse. Thus, in this work, we mainly focus on the scenario where the $\beta$ vector is sparse. However, our results also applies for the high-dimensional scenario where $p\gg n$, and the signal (i.e., the coefficient vector) is sparse to guarantee the feasibility of the system. 
In order to derive the sparsity of the resulting solution, we will consider the regularized version of Problem (\ref{eq:formulation}) with $\ell_1$ penalty on the unknown parameter, which is the Problem (\ref{eq:formulation1}) . Notably, the Problem (\ref{eq:formulation1}) is equivalent to the formulation as follows \cite{rockafellar2015convex},
\begin{eqnarray}
\label{eq:formulationLasso}
\begin{split}
\min_{\beta \in \mathbb{R}^p} & \quad \|\beta\|_1,\\
s.t. &\quad \frac{1}{\sqrt{n}}\|Y - X\beta\|_\infty \leq \lambda_0  ,
\end{split}
\end{eqnarray}
where $\lambda_0 > 0$ is a tuning parameter to control the sparsity of the resulting estimation. Also, $\lambda_0$ and $\lambda$ has the one-to-one correspondence. 
It can be easily checked that Problem (\ref{eq:formulationLasso}) is a convex system and can be recast as a linear programming (LP) problem:
\begin{eqnarray}
\begin{split}
\min_{\gamma \in \mathbb{R}^{+p}} &\quad \sum_{i = 1}^p \gamma_i\\
s.t. &\quad -\gamma_i \leq \beta_i \leq \gamma_i, \mbox{ for } i = 1, \cdots, p\\
&\quad - \lambda_0  \leq \frac{1}{\sqrt{n}}(Y - X\beta)_j \leq \lambda_0 , \mbox{ for } j - 1, \cdots, n.
\end{split}
\end{eqnarray}
For the interest of our proof, we use the constrained version as shown in Problem (\ref{eq:formulationLasso}) in our theoretical analysis.

\subsection{Main results} \label{sec:main}
In this section, we present our main theoretical results on the estimation error of the resulting estimator that is proposed in Formulation (\ref{eq:formulationLasso}). The proofs are postponed to the Appendix.
\subsubsection{Restricted strong convexity assumption}
In this subsection, we will summarize the assumptions we need for deriving the main results. These are standard assumptions used in the statistical literature on recovering the sparse signals.

\begin{definition}
The restricted strong convexity (RSC) condition on model matrix $X$ with respect to a set $\mathcal{C}$ is the following, there exists some constant $\gamma > 0$ such that:
\[\frac{\frac{1}{n}\nu X^TX \nu}{\|\nu\|_2^2} \geq \gamma\text{ for all nonzero } \nu \in \mathcal{C},\]
\end{definition}
\noindent where $\gamma$ is called the restricted eigenvalue bound with regard to $\mathcal{C}$.
\begin{assumption}
\label{assump:RSC}
The restricted strong convexity (RSC) condition holds on the following set:
\[
\mathcal{C} = \left\{\nu \in \mathbb{R}^p \left| \|\nu_{\mathcal{S}^C}\|_1 \leq \|\nu_{\mathcal{S}}\|_1
\right.
\right\}.
\]

\end{assumption}
\noindent We have $\mathcal{C} \subset \mathbb{R}^p$ strictly since it is of the form of a cone.

The RSC (Assumption \ref{assump:RSC}) is a standard assumption in the literature for proving the consistency results of regularized high-dimensional sparse estimation problems.
According to \cite{candes2005decoding, candes2007dantzig}, the assumption on restricted strong convexity condition will hold with high probability for random design matrices $X$, such as a random matrix with i.i.d. Gaussican entries, or Rademacher entries.

%
%
%
%

\subsubsection{Estimation error upper-bound}
Before we give the estimation error upper-bound,  we first prove a lemma, which states that the unknown ground truth $\beta^\ast$, is feasible to Problem (\ref{eq:formulationLasso}) with high probability.
\begin{lemma}\label{lemma:feasibleTruth}
Let $\lambda_0 = \sigma\sqrt{\alpha  \frac{\log p}{n}}$ for some $\alpha > 2 $, then with probability exceeding $1 - \frac{2n}{p} \cdot (\frac{1}{p})^{\frac{\alpha - 2}{2}} $, the unknown ground truth $\beta^\ast$ is a feasible solution to Problem (\ref{eq:formulationLasso}).
\end{lemma}

This lemma basically tells that, if the noise in the data generation is sub-Gaussian, when $\frac{2n}{p} \cdot (\frac{1}{p})^{\frac{\alpha - 2}{2}} \to 0 $, the ground truth $\beta^\ast$ lies in the  feasible region of Problem (\ref{eq:formulationLasso}).
Further, if the noise has a bounded distribution within the interval $[-\lambda_0,\lambda_0]$, $\beta^\ast$ lies in the  feasible region of Problem (\ref{eq:formulationLasso}) with probability 1.

\begin{theorem}
\label{thm:estError_truth_feasible}
Let $\hat{\beta}$ denote the optimal solution to Problem (\ref{eq:formulationLasso}).
Suppose $\beta^\ast$ is any unknown sparse ground truth with $|supp(\beta)| \leq S$, and the  restricted strong convexity assumption holds  for the observed feature matrix $X$ with $\gamma$ in $\mathcal{C}$. Let $\lambda_0 = O(\sigma\sqrt{  \frac{\log p}{n}})$.
We have the following non-asymptotic upper-bound on the estimation error:
\begin{equation}
\label{eq:thm_truth_feasible}
\|\hat{\beta} - \beta^\ast\|_2 \leq O(\sigma\sqrt{S\log p}).
\end{equation}
\end{theorem}

According to the above result, the performance of sparse parameter estimation from Problem (\ref{eq:formulationLasso}) is guaranteed. Not only we will be able to recover the sparse unknown parameter in the high-dimensional scenario, but also we can bound the mean squared error (MES) of the resulting estimation in an order of  
$O(\sigma\sqrt{S\log p})$.
This upper-bound in fact is consistent with the MSE in \cite{candes2007dantzig}.

However, we may obtain a better bound if we add one more condition: the maximum absolute column sum of the matrix norm of the observed feature matrix $X$: $\|X\|_1 := \max_{i = 1}^{p}\|X_i\|_1$ is also in the order of $\sqrt{n}$.
This condition usually holds with sparse matrix, where there are a lot of $0$ entries, or even matrix with columns, where the entry values decay exponentially. This is almost the case for the fuselage shape control problem that the feature matrix $X=(BU_1, -BU_2)$ has a lot of near zero (magnitude of $10^{-6}$) entries due to the holding fixtures of the fuselages. 
Specifically, we state our conclusion in the following theorem:

\begin{theorem}
\label{thm:estError}
Let $\hat{\beta}$ denote the optimal solution to Problem (\ref{eq:formulationLasso}).
Suppose $\beta^\ast$ is any unknown sparse ground truth with $|supp(\beta)| \leq S$, and the  restricted strong convexity  assumption holds  for the observed feature matrix $X$ with $\gamma$ in $\mathcal{C}$. Let $\lambda_0 = O(\sigma\sqrt{  \frac{\log p}{n}})$.
If we further have that $\|X\|_1 = O(\sqrt{n})$, we have the following tighter non-asymptotic upper-bound on the estimation error:
\begin{equation}
\label{eq:thm_truth_feasible}
\|\hat{\beta} - \beta^\ast\|_2 \leq O(\sigma\sqrt{\frac{S\log p}{n}}).
\end{equation}

\end{theorem}


 The proofs to Theorem \ref{thm:estError_truth_feasible} and \ref{thm:estError} are in Appendix \ref{proof:thm:estError_truth_feasible} and \ref{proof:thm:estError}, respectively. 
 The above Theorem \ref{thm:estError_truth_feasible} and \ref{thm:estError} essentially state that with high probability, the optimal solution from our proposed formulation in (\ref{eq:formulationLasso}) is very close to the unknown ground truth, with their difference bounded by $O(\sigma\sqrt{\frac{S\log p}{n}})$.

\subsubsection{Prediction Error Upper-Bound}
\begin{theorem}
\label{thm:predictError}
Under the assumptions of Theorem \ref{thm:estError}, we have the following upper-bound for the prediction error:
\begin{equation}
\|X(\hat{\beta} - \beta^\ast)\|_2 \leq O(\sigma\sqrt{\frac{S\log p}{n}}).
\end{equation}
\end{theorem}
The proof is simple, which is due to the by-product in the proof of Theorem \ref{thm:estError} and is shown in Appendix \ref{proof:thm:predictError}. This theorem tells that even if we are aiming to find a sparse solution to minimize the $\ell_\infty$ loss function, the resulting solution will still produce a good estimation in the sense of the mean squared error.


%

%
\section{Algorithm}\label{sec:alg}
In this section, we proposed an algorithm to efficiently solve the optimization problem (\ref{eq:objforce_l1}) and (\ref{eq:objforce_step2}) based on the ADMM algorithm \cite{boyd2011distributed}.  Compared to the gradient method that is sensitive to the choice of the step-size, the ADMM algorithm is robust to the parameter selection for algorithm convergence \cite{ghadimi2015optimal}. The ADMM algorithm is well suited to the convex optimization problem with separate objectives, which is the case of the optimization problem (\ref{eq:objforce_l1}), so we adopt ADMM algorithm for parameter estimation. Problem (\ref{eq:objforce_l1}) can be further written as  

\begin{equation}
\min_{y} (\sum_{i = 1}^2 f_i(y_i)+g(z)),\\
\end{equation}
where 

$$f_1(y_1)=\|B(\psi_2-\psi_1+U_2F_2-U_1F_1)\|_\infty=\|y_1\|_\infty,$$ 
$$f_2(y_2)=\lambda\|F_1\|_1+\lambda\|F_2\|_1=\lambda\|y_2\|_1,$$ 
$$y_2=[F_1', F_2']',$$
$$y_1=Ay_2+b, A=[-BU_1, BU_2], b=B(\psi_2-\psi_1), y=[y_1', y_2']'.$$
\[
  g(z)=
  \begin{cases*}
                                   \infty, & $Ez\neq b$, \\
                                   	0, & otherwise, 
 \end{cases*}  
\]
$$E=(I, -A), z=y.$$


According to the ADMM formula, the algorithm for the above problem can be written as
\begin{equation}
\ y_i^{k+1}=\mbox{prox}_{f_i}(z_i^{k}-u_i^{k}),\\
\end{equation}
\begin{equation}
\ z^{k+1}=\mbox{prox}_{g}(y^{k+1}+u^{k}),\\
\end{equation}
\begin{equation}
\ u^{k+1}=y^{k+1}+u^{k}-z^{k+1}.\\
\end{equation}
Here, prox$_f (x)$ is the proximal operator of $f$, which defined as 
\begin{equation}
\mbox{prox}_{f}(x)={\arg \min_{u}\left( f(u)+1/2 \|u-x\|_2^2 \right)}.\\
\end{equation}
The proximal operator prox$_{f_i}$ and prox$_g$ can be analytically derived if we formulated the optimization problem in this way.   From Chapter 6 of \cite{beck2017first}, the proximal operators of functions ${f_i}(y_i)$ and $g(z)$ can be derived as 
\begin{equation}
\mbox{prox}_{f_1}(y_1)=y_1-P_{B_ {\|\cdot\|_1 [0,1]}}(y_1),
\end{equation}
\begin{equation}
\mbox{prox}_{f_2}(y_2)=S_{\lambda/\rho}(y_2),
\end{equation}
\begin{equation}
\mbox{prox}_{g}(z)=z-E'(EE')^{-1}(Ez-b),
\end{equation}
where $P_{B_ {\|\cdot\|_1 [0,1]}}(\cdot)$ is the projection onto the unit ball $G=B_ {\|\cdot\|_1 [0,1]}=\{x\in \mathbb{R}^n:\|x_1\|_1\leq1\}$, and 
\[
  P_{B_{\|\cdot\|_1}[0,1]}(y_1)=
  \begin{cases*}
                                   y_1, & $\|y_1\|_1 \leq 1$, \\
                                   	S_{\lambda^\ast}(y_1), & $\|y_1\|_1 > 1$,
 \end{cases*}
\]
where $\lambda^*$ is any positive root of the non-increasing function $\varphi(\lambda)=\|S_\lambda(y_1)\|_1-1$. $S_\lambda$ is the soft thresholding operator of $\|\cdot\|_1$, which is 
\begin{equation}
S_\lambda(v)=(v-\lambda)_+ - (-v-\lambda)_+,
\end{equation}
where $(x)_+$ is short for $\max\{x,0\}$. 

Due to the convexity of Problem (\ref{eq:objforce_l1}), the convergence of the ADMM is guaranteed given a reference point and we can obtain the global optimum of the optimization Problem (\ref{eq:objforce_l1}). In practice, Boyd et al.  \cite{boyd2011distributed} suggested that a reasonable termination of the ADMM is that both the primal and dual residuals should be small, i.e., 
\begin{equation}
\|r^k\|_2\leq e_1,  \|s^k\|_2\leq e_2,
\end{equation}
where $r^k=y^k-z^k$ is the primal residual, and $s^k=-\rho(z^k-z^{k-1})$  is the dual residual. $e_1>0$ and $e_2>0$ are the tolerances, which can be chosen by using an absolute criterion and a relative criterion, i.e., we have 
\begin{equation}
e_1=\sqrt{m}e_3+e_4\max\{\|z^k\|_2,\|y^k\|_2\},
\end{equation}
\begin{equation}
e_2=\sqrt{m}e_3+e_4\|\rho u^k\|_2,
\end{equation}
where $m=m_1+m_2$, $e_3>0$ and $e_4>0$ are an absolute tolerance and a relative tolerance, respectively. The following algorithm is proposed for estimations of $F_1$ and $F_2$, and $K$ is the maximum number of iterations.

\begin{algorithm}
\caption{The ADMM of estimating $y_2=[F_1', F_2']'$ }\label{alg:admm1}
\begin{algorithmic}[1]
\STATE Input: $\rho, \lambda, e_3, e_4, m_1,m_2, n, B, U_1, \psi_1, U_2, \psi_2, A, b$
\STATE Initialize $z^0, u^0$
\STATE \emph{loop}:
\FOR {$k=1: K$}
\STATE $u_1^k=\{u_i^k\}_{i=1,...,2n},u_2^k=\{u_i^k\}_{i=1+2n,...,2n+m_1+m_2}$,
\STATE $z_1^k=\{z_i^k\}_{i=1,...,2n},z_2^k=\{z_i^k\}_{i=1+2n,...,2n+m_1+m_2}$,
\STATE $y_1^{k+1}=z_1^k-u_1^k-P_{B_ {\|\cdot\|_1 [0,1]}}(z_1^k-u_1^k)$,
\STATE $y_2^{k+1}=S_{\lambda/\rho}(z_2^k-u_2^k)$,
\STATE $y^{k+1}=[(y_1^{k+1})', (y_2^{k+1})']'$,
\STATE $z^{k+1}=y^{k+1}+u^k-E'(EE')^{-1}(Ez-b)$,
\STATE $u^{k+1}=u^k+y^{k+1}-z^{k+1}$,  
\STATE $r^{k+1}=y^{k+1}-z^{k+1}$,
\STATE $s^{k+1}=\rho(z^{k+1}-z^k)$,
\STATE $e_1=\sqrt{m}e_3+e_4\max\{\|z^{k+1}\|_2,\|y^{k+1}\|_2\}$,
\STATE $e_2=\sqrt{m}e_3+e_4\|\rho u^{k+1}\|_2$,
\IF{$\|r^{k+1}\|_2\leq e_1$, and $\|s^{k+1}\|_2\leq e_2$}
\RETURN $y_2$
\ENDIF
\ENDFOR
\end{algorithmic}
\end{algorithm}

Since we have the sparsity requirement of $\|F_1\|_0+\|F_2\|_0\leq M$, the $\lambda$ value needs some investigations. Binary search can be used for computing for this value to meet the aforementioned requirement. The lower bound of the search range is trival, which is 0, but the upper bound needs some theoretical investigations. In the following Proposition \ref{prop:maxLambda}, we provide an upper bound of the search range. 
\begin{proposition}
\label{prop:maxLambda}
If we choose the tuning parameter $\lambda$ 
in the Problem (\ref{eq:objforce_l1}) 
such that $\lambda > \|X\|_\infty$, $X=[BU_1,-BU_2]$, there exists one unique global minimum to  (\ref{eq:objforce_l1}) 
, with $\hat{\beta}=[F_1', F_2']'= 0$.
\end{proposition}
The main idea to prove this proposition is that at any point that is not the origin, we can find a direction, along which, the objective value decreases.
Due to space limit, we postpone the proof into Appendix \ref{proof:prop:maxLambda}.

According to  Proposition \ref{prop:maxLambda}, we can set $ \|X\|_\infty$ as the upper bound of the search range to find the  $\lambda$ value that meets the requirement $\|F_1\|_0+\|F_2\|_0=M$. The nonzero components of $F_1$ and $F_2$ correspond to the optimal locations of actuators for shape control of fuselage 1 and fuselage 2, respectively. Then, the forces from each actuator can be found by solving the optimization problem (\ref{eq:objforce_step2}). Similarly, the ADMM algorithm can also be used in (\ref{eq:objforce_step2}) to get the optimal solution by the Algorithm \ref{alg:admm2} in Appendix \ref{add2}.

\section{Case study} \label{sec:numeric}


In this section, we use the fuselage assembly process to validate the proposed methodology. 
We  present our detailed experiment setting and report our results in Subsections \ref{subsec:numeric} and \ref{subsec:simulation}, respectively.

\subsection{Numerical setting} \label{subsec:numeric}
In this case study, we use an FEA model \cite{wen2018feasibility} to generate data and validate the proposed methodology. The FEA model used in this case study has been validated with the experimental data, and more details about FEA model can be found in \cite{wen2018feasibility, yue2018surrogate}. For shape adjustments of single fuselage, the method from \cite{du2019fuselage} achieves the best control performance in terms of reduction of shape deviations after adjustments and the forces applied during shape adjustments. 

In order to make a fair comparison between our methods and the best result of composite fuselage shape control in current literature, we use the same FEA data from \cite{du2019fuselage}. 
There are 20 incoming fuselages with different dimensional variations. Based on engineering practice, in each pair of the fuselages, $m_1=m_2=18$ feasible actuators are equally placed from -12 degrees to 192 degrees in the lower part of each fuselage. More details about the data generation can be found in \cite{du2019fuselage}.  $M=M_1+M_2=18$  actuators are used for shape adjustments of two fuselages in this case study. The number of measurement points along the interface of the pair of fuselages are $n_1=n_2=n=182$, weight matrix is set to be $B=\mbox{diag}(1/n)$, and $U_1=U_2$ are chosen to be the same with \cite{du2019fuselage} for comparison purpose. 

In practice, some elements of displacement matrix $U_i,i=1,2,$ are very small due to the structures of composite fuselages, such as the elements of $U_i$ near the fixture. Specifically, we checked our displacement matrix $U_i$  and find that the element scales decays almost exponentially for each column, which satisfies our column norm assumption in Theorem \ref{thm:estError}.
To achieve a better computational performance, such as reducing floating point errors \cite{demmel1997applied}, we multiply a large constant number $L_N$ in the objective function $ \|B(\psi_2+U_2F_2-(\psi_1+U_1F_1))\|_\infty $ and $ \|B(\psi_2+{U_2}^c{F_2}^c-(\psi_1+{U_1}^c{F_1}^c))\|_\infty $ in the optimization Problems (\ref{eq:objforce_l1}) and (\ref{eq:objforce_step2}), which does not change the optimal solution. In this way, the computational problems induced by $U_i$ matrix can be alleviated in real implementations. In this case study, we set $L_N=10^7$. In terms of the ADMM algorithm, we set $\rho=1$, $e_3=10^{-6}$,  $e_4=10^{-5}$.

\subsection{Results of the proposed method}\label{subsec:simulation}
In this case study, we randomly pick up two different fuselages from 20 incoming fuselages and have 50 replications. Consequently, we have 50 pairs of fuselages for assembly. According to the engineering practice, we use root mean square gap (RMSG), maximum gap (MG), and maximum force (MF) to evaluate the control results. These quantities are defined as follows: 
\begin{equation}
RMSG := \frac{1}{n}\sqrt{\Delta'\cdot\Delta},
\end{equation}

\begin{equation}
MG := \sqrt{\|\Delta_Y \ast\Delta_Y+\Delta_Z \ast\Delta_Z\|_\infty},
\end{equation}


\begin{equation}
MF_i := \|F_i\|_\infty, i = 1, 2,
\end{equation}
where `$\ast$' indicates the Hadamard product, and $\Delta_Y\in \mathbb{R}^ n$ and $\Delta_Z\in \mathbb{R}^ n$ are the first and last $n$ elements of $\Delta$, respectively. Thus,  $MG$ captures the maximum gap for two fuselage assembly. Notably, we only evaluate the $MF_i$ for the safety purpose during the fuselage shape control, since large force may destroy the fuselage structure during fuselage shape control. 
The control results of the proposed method on these 50 pairs of fuselages assembly are shown in Table \ref{table:control}. Since we care more about the maximum gap after control and maximium force used for shape control in fuselage assebly process, we also listed the  maximum (max) of RMSG, MG and MF for these 50 pairs. The results show that the proposed method perfoms well in terms of the maximum gap by using relatively smaller forces, which is acceptable for fuselage assembly in practice. The following Figure \ref{fig:gap_control} and Figure \ref{fig:max_control}  show the box-plots of the fuselage gap after control and the maximum forces used for shape control in these 50 pairs. 

\begin{table}[htp]
\begin{center}
\begin{tabular}{|c|c|c|c|c|}
\hline
~&RMG (inches)&MG (inches)&$MF_1$ (lbs)&$MF_2$ (lbs)\\
\hline
Mean&0.0014&0.0034&304.73&307.42\\
Max&0.0023&0.0060&553.83&776.86\\
Std&0.0004&0.0011&101.43&131.01\\
\hline
\end{tabular}
\end{center}
\caption{Control results of our method on 50 pairs of fuselages. }
\label{table:control}
\end{table}%

\begin{figure}[htbp]
    \centering
    \begin{subfigure}[b]{0.48\textwidth}
        \includegraphics[width=\textwidth]{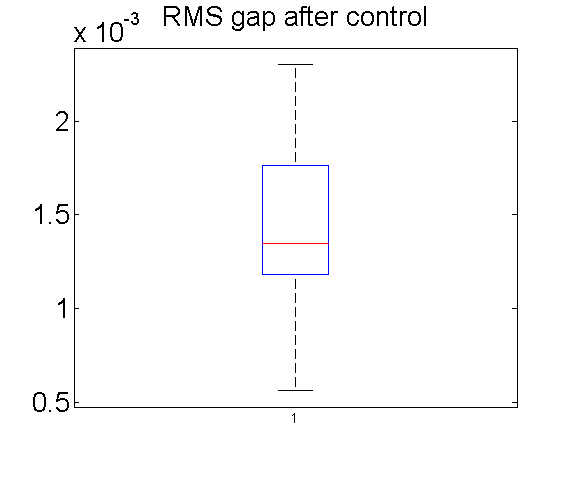}
    \end{subfigure}
    ~ 
    \begin{subfigure}[b]{0.48\textwidth}
        \includegraphics[width=\textwidth]{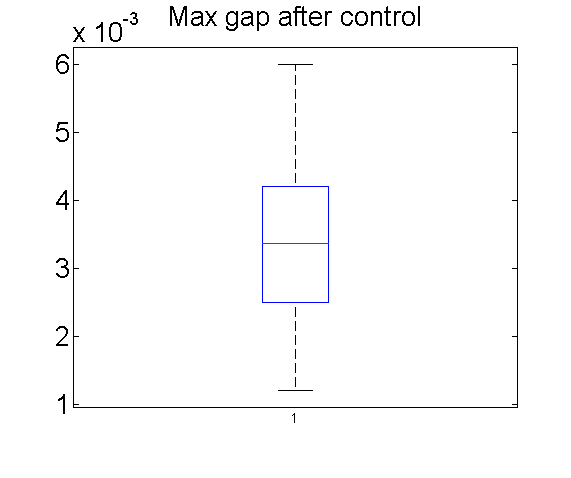}
    \end{subfigure}
    \caption{The boxplots of RMS gap and Max gap of 50 pairs of fuselages after control by the proposed method.}
    \label{fig:gap_control}
\end{figure}

\begin{figure}[htbp]
    \centering
    \begin{subfigure}[b]{0.48\textwidth}
        \includegraphics[width=\textwidth]{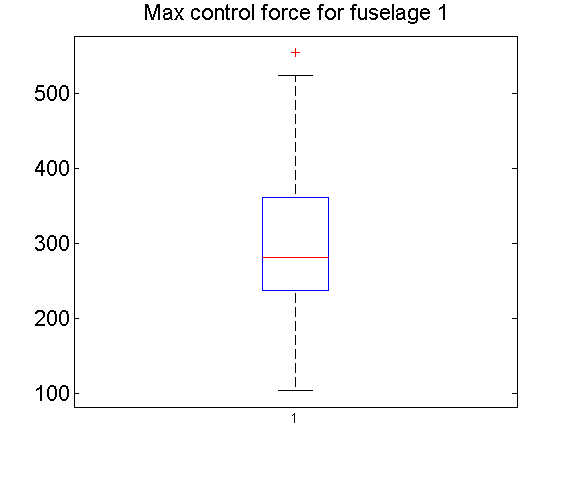}
    \end{subfigure}
    ~ 
    \begin{subfigure}[b]{0.48\textwidth}
        \includegraphics[width=\textwidth]{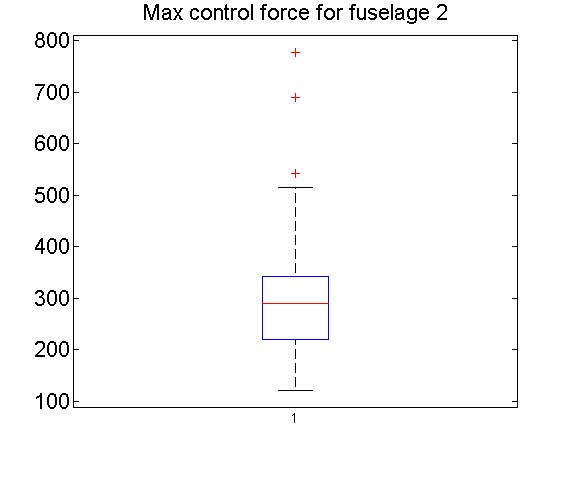}
    \end{subfigure}
    \caption{The boxplots of max control forces for fuselage 1 and fuselage 2 by the proposed method.}
    \label{fig:max_control}
\end{figure}

\subsubsection{Comparisons}\label{subsec:compare}
In this subsection, we compare the proposed method with the current practice \cite{du2019fuselage}, which is to force the incoming fuselage into the design shape in terms of the $\ell_2$ loss, and then assemble. We randomly choose one pair of fuselages to compare the performance of the proposed method with the one in \cite{du2019fuselage}. Figure \ref{fig:dimension1} shows the incoming fuselage dimensions at the interface between two fuselages and the proposed acutator locations of fuselage 1 and fuselage 2. Similarly, we also show the incoming fuselage dimensions and the acutator locations from \cite{du2019fuselage}  of fuselage 1 and fuselage 2 in Figure  \ref{fig:dimension2}. To show the incoming dimensional variations more clearly, we magnify the distortions 150 times since the dimensional error is much smaller than the dimension of the fuselage. The marked round points are the locations of actuators. As shown in the Figure \ref{fig:dimension1} and Figure \ref{fig:dimension2}, the optimal actuator locations of the proposed method are different from  \cite{du2019fuselage} . This is because our method aims to minimize the maximium gap between the pair of assembly fuselages, while \cite{du2019fuselage} aims to control single fuselage into design shape and then assembly. 
We also show the shape adjustments of the proposed method and the one in  \cite{du2019fuselage} for this pair of fuselages in Figure  \ref{fig:adjustment1} and Figure  \ref{fig:adjustment2}, respectively. The black dots show the intial dimensional gap between two fuselages in Y and Z directions; the green dots and blue dots show the adjustment from fuselage 1 and fuselage 2 in Y and Z directions; the red dots are the final adjusted gap after control in Y and Z directions. As shown in Figure  \ref{fig:adjustment1}, the shape adjustments of our method reduces the intial gap and successfully adjusts the shape. Compared to our method, the method from \cite{du2019fuselage} adjusts the shape of incoming fuselage seperately (shown in Figure  \ref{fig:adjustment2}) since the goal of this method is to adjust the shape of the fuselage into the design shape. The control performance is not as good as the proposed method since we can see larger variations around 0 in Figure  \ref{fig:adjustment2}. This pair of fuselages illustrates the effectiveness of the proposed method in detail.

\begin{figure}[htbp]
    \centering
    \begin{subfigure}[b]{0.48\textwidth}
        \includegraphics[width=\textwidth]{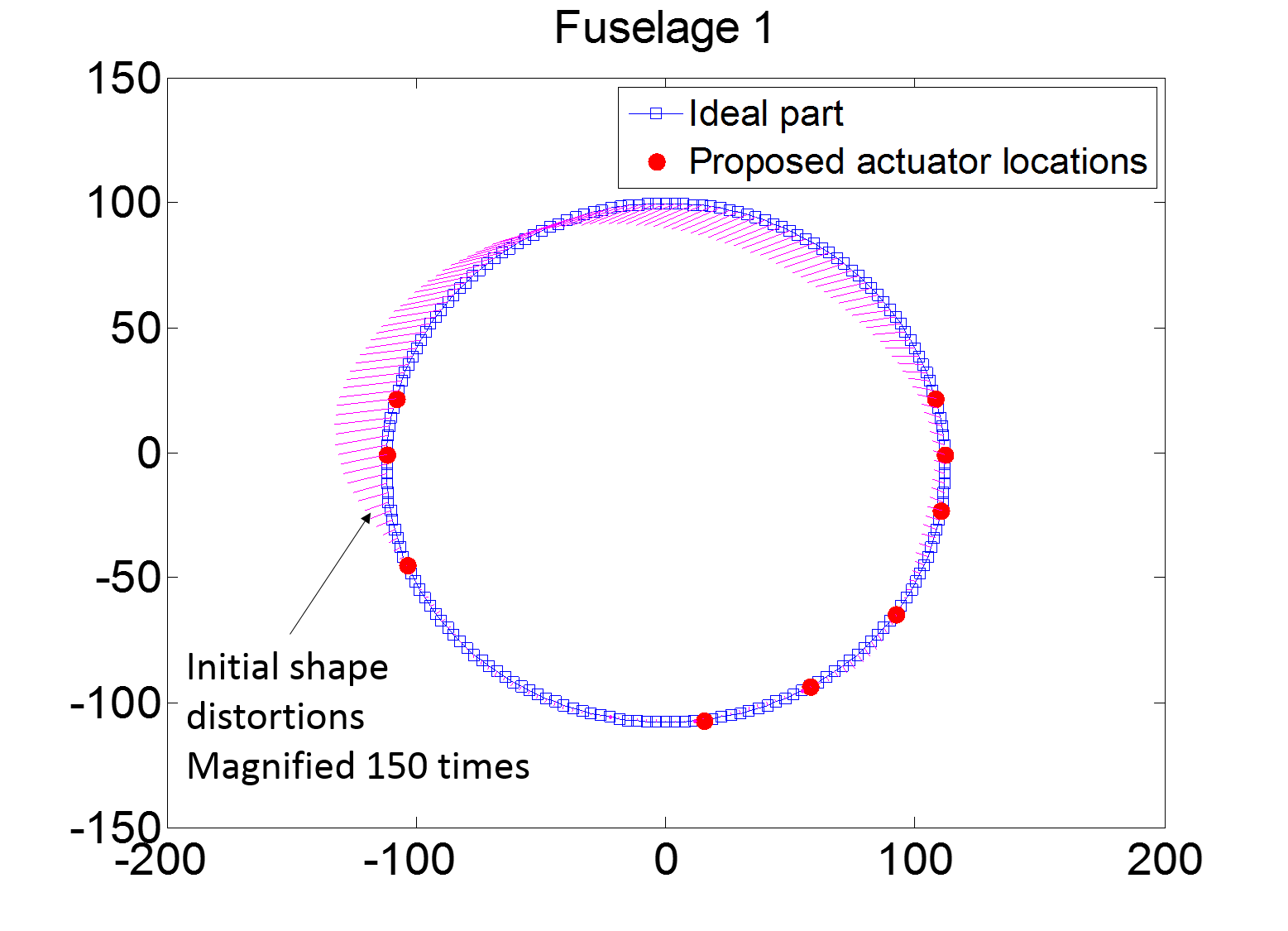}
    \end{subfigure}
    ~ 
    \begin{subfigure}[b]{0.48\textwidth}
        \includegraphics[width=\textwidth]{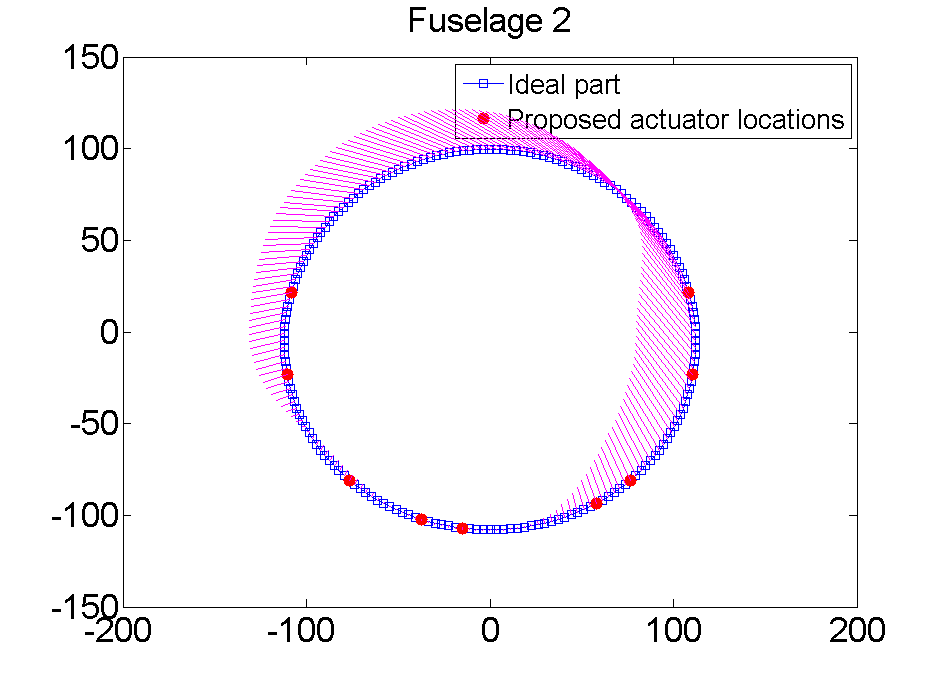}
    \end{subfigure}
    \caption{Incoming fuselage dimensions and the proposed actuator locations for fuselage 1 and fuselage 2.}
    \label{fig:dimension1}
\end{figure}

\begin{figure}[htbp]
    \centering
    \begin{subfigure}[b]{0.48\textwidth}
        \includegraphics[width=\textwidth]{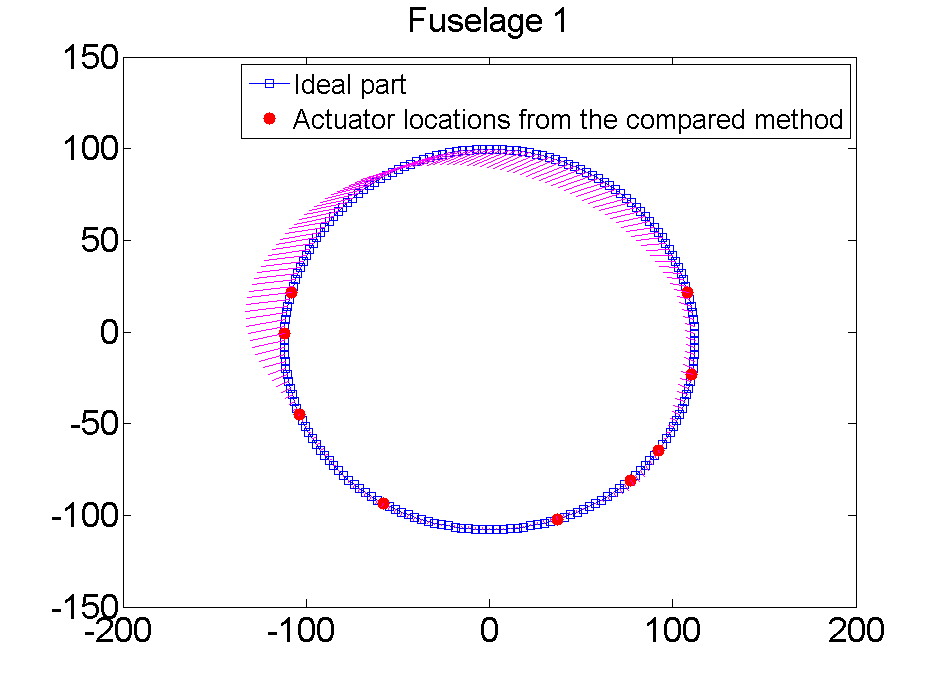}
    \end{subfigure}
    ~ 
    \begin{subfigure}[b]{0.48\textwidth}
        \includegraphics[width=\textwidth]{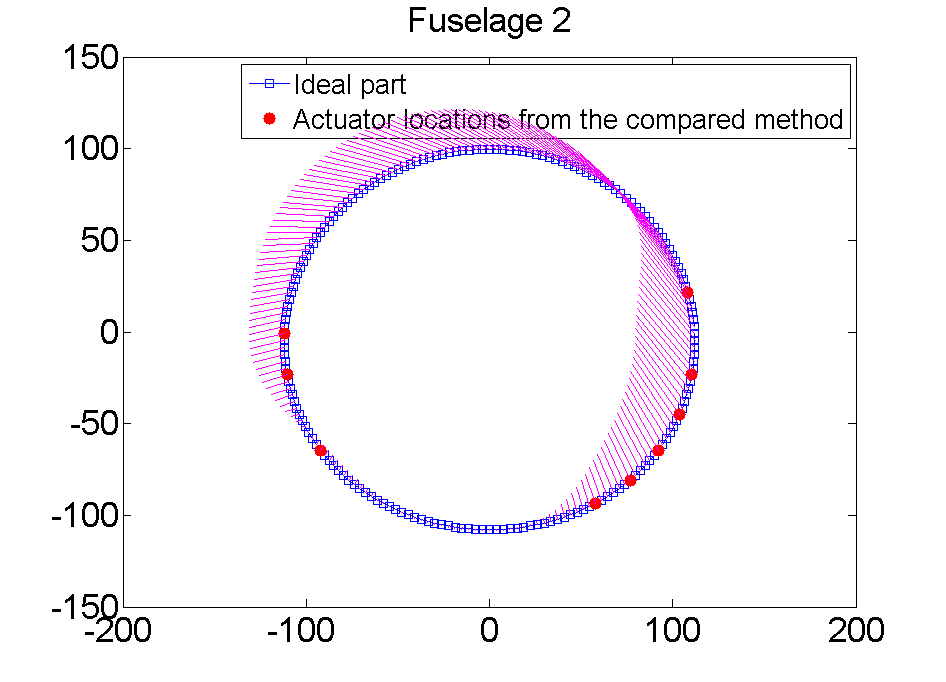}
    \end{subfigure}
    \caption{Incoming fuselage dimensions and the actuator locations from \cite{du2019fuselage} for fuselage 1 and fuselage 2.}
    \label{fig:dimension2}
\end{figure}

In order to better illustrate the improvement of the proposed method, we calculated the improvements on the gap after shape adjustments, and use statistical t test to test the significance of the improvement on the gap. The improvement results on mean and standard deviation (std) of the gap after adjustments are listed in Table \ref{table:gap_reduction_MG}. Notably, the  improvement means the control results via current practice \cite{du2019fuselage} minus the results of the proposed method. The null hypothesis $H_0$: the improvement comes from a distribution with mean zero, i.e., no significant improvement. The alternative hypothesis $H_1$: the mean of improvement is greater than 0 (right-tailed test), which indicates the proposed method reduces the dimensional gap after control. The $p-$value is also listed in Table \ref{table:gap_reduction_MG}. As shown in Table \ref{table:gap_reduction_MG}, compared with  \cite{du2019fuselage}, the proposed method significantly reduces the maximum gap (MG) and RMS gap (RMSG) for two fuselage assembly. The box-plots of MG and RMSG improvement are shown in Figure \ref{fig:gap_improve}.  



\begin{figure}[htbp]
    \centering
    \begin{subfigure}[b]{0.38\textwidth}
        \includegraphics[width=\textwidth]{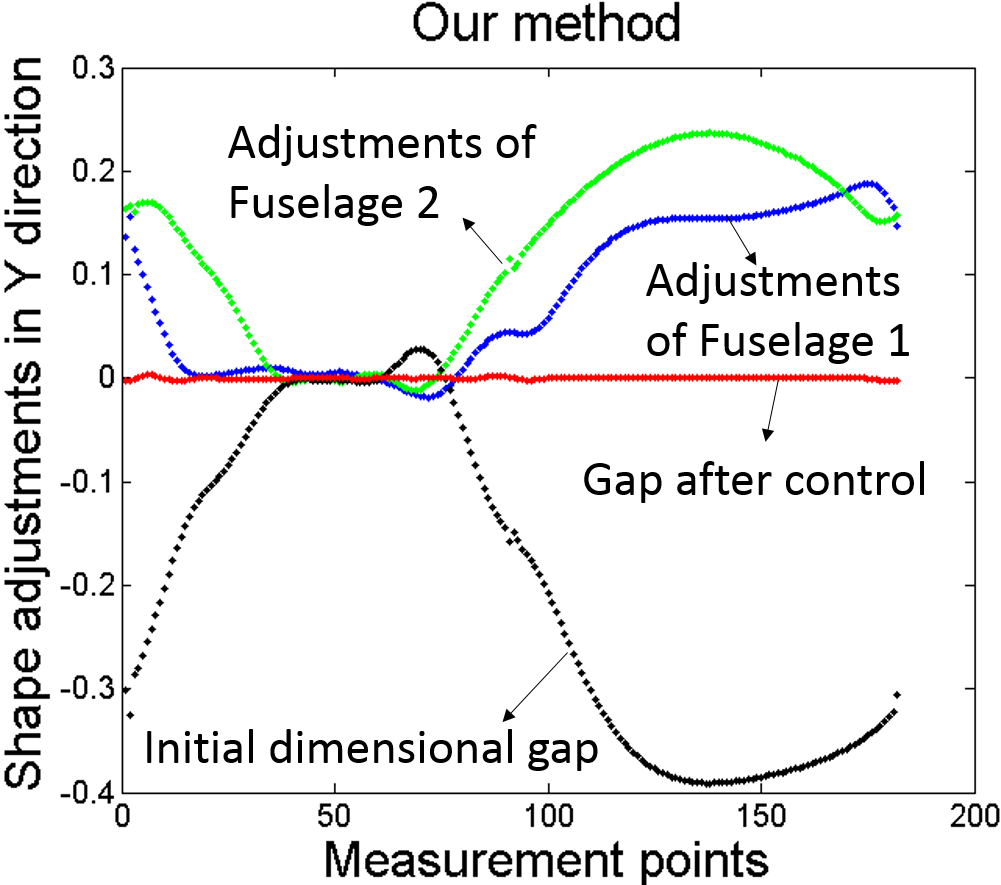}
    \end{subfigure}
    ~ 
    \begin{subfigure}[b]{0.38\textwidth}
        \includegraphics[width=\textwidth]{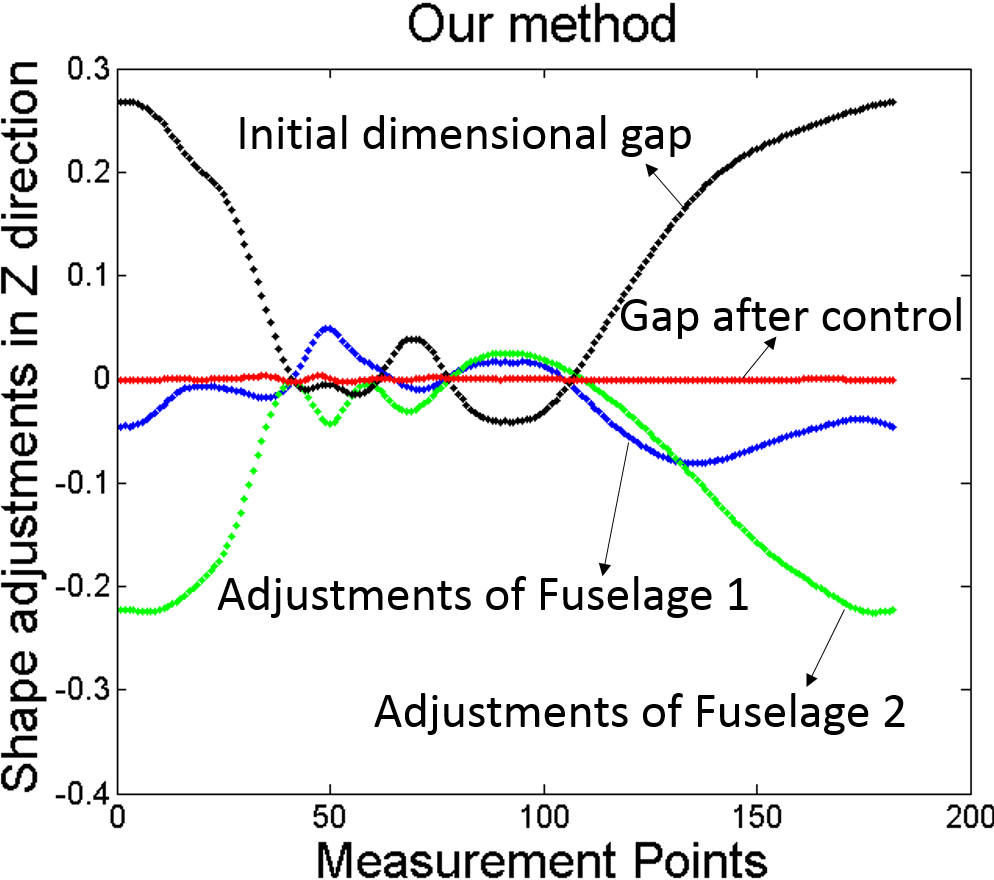}
    \end{subfigure}
    \caption{Incoming fuselage dimensions and the actuator locations from \cite{du2019fuselage} for fuselage 1 and fuselage 2.}
    \label{fig:adjustment1}
\end{figure}

\begin{figure}[htbp]
    \centering
    \begin{subfigure}[b]{0.38\textwidth}
        \includegraphics[width=\textwidth]{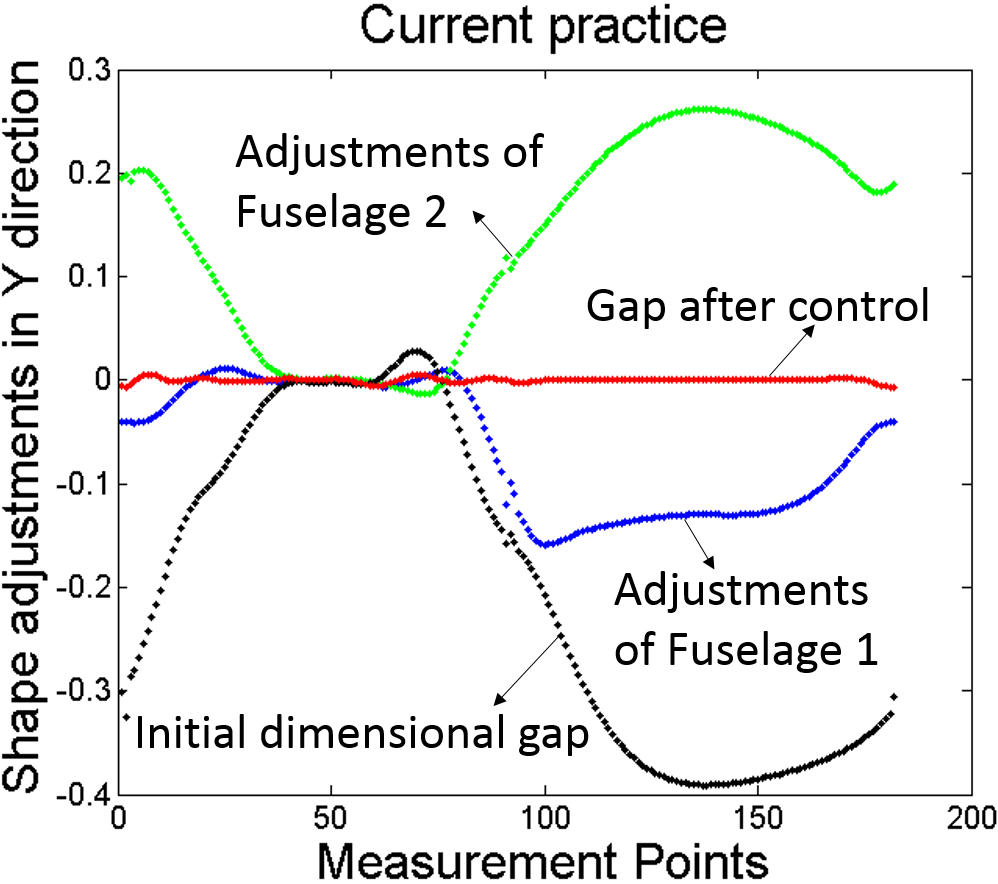}
    \end{subfigure}
    ~ 
    \begin{subfigure}[b]{0.38\textwidth}
        \includegraphics[width=\textwidth]{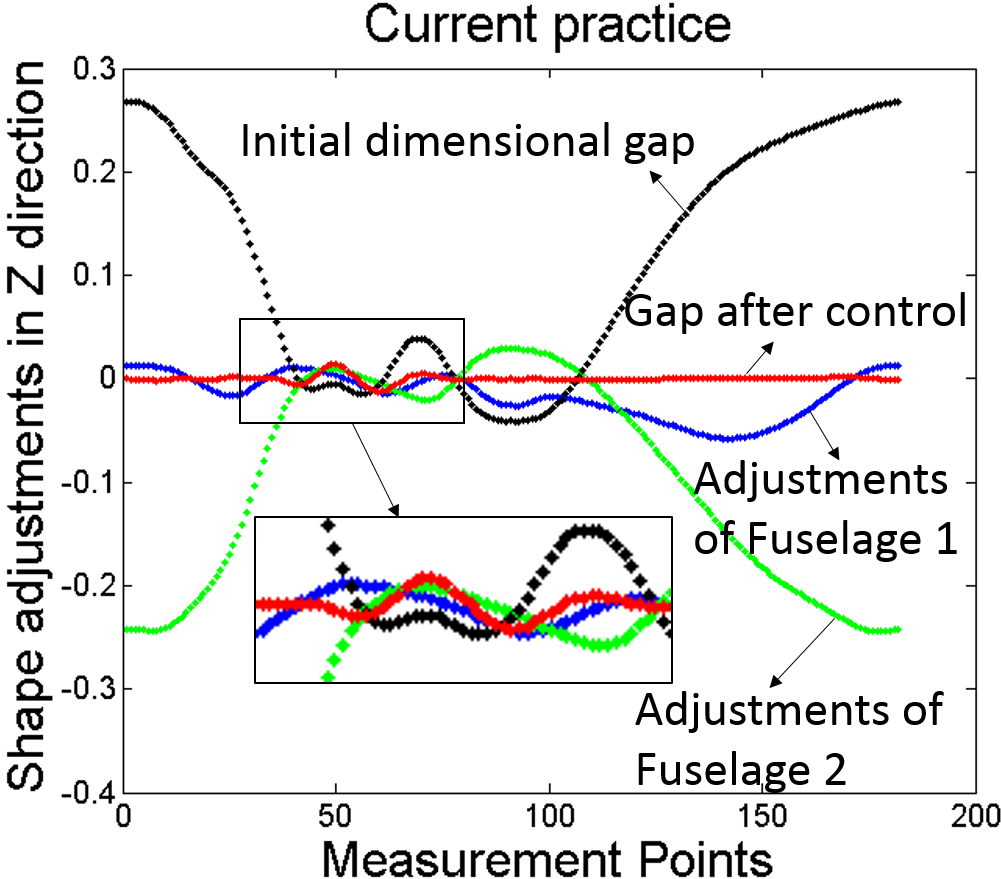}
    \end{subfigure}
    \caption{Incoming fuselage dimensions and the actuator locations from \cite{du2019fuselage} for fuselage 1 and fuselage 2.}
    \label{fig:adjustment2}
\end{figure}

\begin{table}[htp]
\begin{center}
\begin{tabular}{c|c|c|c|c|c}
\hline
\multicolumn{3}{c}{MG Improvement (inches)} &\multicolumn{3}{c}{RMSG Improvement  (inches)}\\
\hline
Mean &Std&P-value &Mean &Std&P-value \\
0.0154&0.0074&$6.19 \times 10^{-20}$&0.0048&0.0023&$5.40\times 10^{-20}$\\
\hline
\end{tabular}
\end{center}
\caption{Gap reduction of our method compared with method from \cite{du2019fuselage}. }
\label{table:gap_reduction_MG}
\end{table}

\begin{figure}[htbp]
    \centering
    \begin{subfigure}[b]{0.48\textwidth}
        \includegraphics[width=\textwidth]{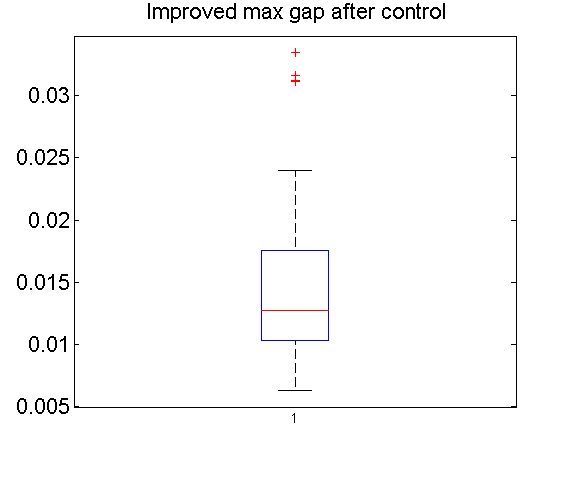}
    \end{subfigure}
    ~ 
    \begin{subfigure}[b]{0.48\textwidth}
        \includegraphics[width=\textwidth]{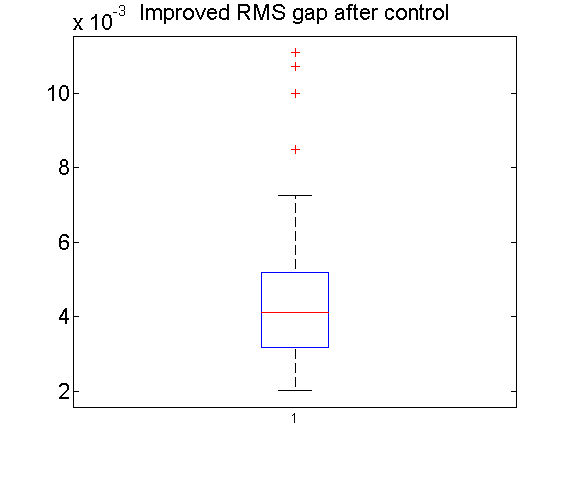}
    \end{subfigure}
    \caption{The boxplots of improved max and RMS gap after control compared with method from \cite{du2019fuselage}.}
    \label{fig:gap_improve}
\end{figure}
%


\subsection{Discussions}\label{subsec:discussion}
A contribution in this paper is that we  propose an optimal shape control strategy for composite fuselage assembly. In practice, the main concern during the composite fuselage assembly process is the maximum of the dimensional gap. Consequently, the objective of this paper is to minimize the maximum gap in composite fuselage shape control instead of using the square error as in \cite{du2019fuselage}. The goal of \cite{du2019fuselage} is to achieve the optimal actuator placement for shape control of single fuselage, and the target shape after shape control is the design shape. In the case study, our method performs much better than the method in \cite{du2019fuselage} in terms of the maximum gap after the shape control. Also, by considering the initial gap between the pair of fuselages, the resulting RMSG is also smaller than those in \cite{du2019fuselage}.

\section{Conclusions}\label{sec:conclusion}
This paper proposes an optimal shape control strategy for composite fuselage assembly. Due to natural dimensional variations of fuselages, there is a gap on the interface of two fuselages before assembly. The current practice adjusted the shape of each fuselage to the design shape in terms of the $\ell_2$ loss and then assemble, which is not optimal for the case of two fuselage assembly. Our contribution is to consider the initial gap of the pair of incoming fuselages and propose a sparse learning model, which aims to minimize the maximum gap ($\ell_\infty$ loss) after shape control. 
From the statistical view, the proposed model is $\ell_1$ sparsity-driven  penalized $\ell_\infty$ loss linear regression.
Under the linear model with light tailed sub-Gaussian errors assumption, we provide the non-asymptotic upper-bound of the estimation error, measurement error, and prediction error.
The real case studies of composite fuselage assembly validate the proposed method according to its control performance. 
Compared with the current literature, the case study shows that our method achieves significant reduction of the maximum gap after control. Notably, although our method is demonstrated for optimal shape control in composite fuselage assembly process, the methodology can be extended for optimal shape control of assembly process in other structures. 

\bibliographystyle{plain}

\bibliography{refs}

\newpage
\section*{Appendix}\label{sec:appendix}
All the proofs in this paper are provided in this appendix. The proofs of main results in Section \ref{sec:main} are in Appendix \ref{sec:A}, and the Algorithm \ref{alg:admm2} in Section \ref{sec:alg} is provided in Appendix \ref{add2}. Specifically, the proofs of Lemma \ref{lemma:feasibleTruth}, Theorem \ref{thm:estError_truth_feasible} and Theorem \ref{thm:estError} are given in \ref{sec:A1}, \ref{proof:thm:estError_truth_feasible} and \ref{proof:thm:estError}, respectively. 

\appendix
\section{Proofs in Section \ref{sec:main}}\label{sec:A}

Throughout this proof, we will assume that the variance parameter in Formulation (\ref{eq:dataGeneration}): $\sigma^2 = 1$. We will also assume that the observed feature matrix $X$ is standardized with column mean as $0$ and standard deviation in the order of $\sqrt{n}$.
The general case would follow from a simple rescaling procedure.
Denote $\mathcal{S}_0 := \mbox{supp}(\beta^\ast)$, the support set of the ground truth $\beta^\ast$.

\subsection{Proof of Lemma \ref{lemma:feasibleTruth}}\label{sec:A1}

\begin{proof}
Given the data generation mechanism in (\ref{eq:dataGeneration}), the random error vector $\epsilon$ is independent sub-Gaussian random variables.
We thus have the following tail inequality:
\begin{equation*}
\mathbb{P}(\frac{\epsilon_i}{\sqrt{n}} > t) \leq 2e^{-\frac{nt^2}{2}}.
\end{equation*}
By Bonferroni bound, we have
\begin{equation*}
\mathbb{P}(\|\frac{\epsilon_i}{\sqrt{n}}\|_\infty > t) \leq 2ne^{-\frac{nt^2}{2}}.
\end{equation*}
Let $t = O(\sqrt{\frac{\log p}{n}})$, say $\lambda_0 = t = \sqrt{\alpha\frac{\log p}{n}}$ for some $\alpha > 2 $,  we have with probability larger than $1 - \frac{2n}{p} \cdot (\frac{1}{p})^{\frac{\alpha - 2}{2}} $, 
\begin{equation*}
|\frac{\epsilon_i}{\sqrt{n}}| = \frac{1}{\sqrt{n}}|Y_i - \langle(X^T)_i,\beta^\ast\rangle| \leq \lambda_0,
\end{equation*}
for all $i = 1, \cdots, n$.

\end{proof}

\subsection{Proof of Theorem \ref{thm:estError_truth_feasible}}
\label{proof:thm:estError_truth_feasible}
\begin{proof}
In order to prove the statement, we will need the following lemma.
\begin{lemma}\label{lemma:nu}
Denote $\nu = \hat{\beta} - \beta^\ast$, the difference vector of the optimal solution and the ground truth. We have $\nu \in \{x \in \mathbb{R}^p: \|x_{\mathcal{S}_0^c}\|_1 \leq \|x_{\mathcal{S}_0}\|_1\}$, which is a cone in $\mathbb{R}^p$. 
\end{lemma}

\begin{proof}
Due to the definition of $\hat{\beta}$, which is the optimal solution to (\ref{eq:formulationLasso}), and $\beta^\ast$ is a feasible solution to (\ref{eq:formulationLasso}), we have 
$$\|\hat{\beta}\|_1 \leq \|{\beta}^\ast\|_1.$$
Thus, we have
\begin{equation*}
\begin{split}
\|{\beta}^\ast\|_1& \geq \|\hat{\beta}\|_1\\
 &= \|\beta^\ast + \nu\|_1\\
& =  \|(\beta^\ast + \nu)_{\mathcal{S}_0} \|_1 +  \|(\beta^\ast + \nu)_{\mathcal{S}_0^c}\|_1\\
& \geq \|\beta^\ast_{\mathcal{S}_0} \|_1 - \| \nu_{\mathcal{S}_0} \|_1 + \|\nu_{\mathcal{S}_0^c}\|_1\\
& \geq \|{\beta}^\ast\|_1  - \| \nu_{\mathcal{S}_0} \|_1 + \|\nu_{\mathcal{S}_0^c}\|_1,
\end{split}
\end{equation*}
which proves the lemma.
\end{proof}

According to the restricted eigenvalue assumption, we have the following inequality:
\begin{equation}
\gamma \|\nu\|_2^2 \leq \frac{1}{n} \|X\nu\|^2_2.
\end{equation}
On the other hand, 
\begin{equation}
\begin{split}
\frac{1}{n} \|X\nu\|^2_2 &= \frac{1}{n}\nu^TX^TX\nu\\
&\leq \frac{1}{n}\|X^TX\nu\|_\infty\|\nu\|_1.
\end{split}
\end{equation}
Since $\|\nu_{\mathcal{S}_0^c}\|_1 \leq \| \nu_{\mathcal{S}_0} \|_1$, and according to Cauchy--Schwarz inequality, we have 
$$\|\nu_{\mathcal{S}_0^c}\|_1 \leq \| \nu_{\mathcal{S}_0} \|_1 \leq \sqrt{S}\| \nu_{\mathcal{S}_0} \|_2.$$
Thus, we have 
\begin{equation}
\begin{split}
\|\nu\|_1 &= \|\nu_{\mathcal{S}_0^c}\|_1 + \| \nu_{\mathcal{S}_0} \|_1 \\
&\leq 2 \| \nu_{\mathcal{S}_0} \|_1\\
&\leq 2 \sqrt{S}\| \nu_{\mathcal{S}_0} \|_2\\
&\leq  2 \sqrt{S}\| \nu\|_2.
\end{split}
\end{equation}
For $\frac{1}{n}\|X^TX\nu\|_\infty$, we have 
\begin{equation}
\begin{split}
\frac{1}{n}\|X^TX\nu\|_\infty &= \max_{i = 1}^p|\langle \frac{1}{\sqrt{n}}X_i, \frac{1}{\sqrt{n}}X\nu \rangle|\\
&\leq \max_{i = 1}^p\|\frac{1}{\sqrt{n}}X_i\|_1\|\frac{1}{\sqrt{n}}X\nu\|_\infty\\
&\leq \|\frac{1}{\sqrt{n}}X\|_1\|\frac{1}{\sqrt{n}}X\nu\|_\infty.
\end{split}
\end{equation}
According to our Formulation (\ref{eq:formulationLasso}),  and the fact that both $\hat{\beta}$ and $\beta^\ast$ are feasible to (\ref{eq:formulationLasso}), we have 
$$\|\frac{1}{\sqrt{n}}X\nu\|_\infty \leq 2\lambda_0.$$
According  to the normalization on $X$ and Cauchy--Schwarz inequality, we have 
\begin{equation}
\label{eq:proof}
 \|\frac{1}{\sqrt{n}}X\|_1 \leq O(\sqrt{n}).
\end{equation}
Combing the upper-bounds above, we have
\begin{equation}
\begin{split}
\gamma \|\nu\|_2^2 \leq \frac{1}{n} \|X\nu\|^2_2 &\leq \frac{1}{n}\|X^TX\nu\|_\infty\|\nu\|_1\\
&\leq O(\lambda\sqrt{Sn})\| \nu\|_2.
\end{split}
\end{equation}
Thus we obtain that 
\begin{equation}
\begin{split}
\|\nu\|_2 \leq O(\sqrt{S\log p}).
\end{split}
\end{equation}
\end{proof}

\subsection{Proof of Theorem \ref{thm:estError}}
\label{proof:thm:estError}
\begin{proof}
Recall the way we bound $\|\frac{1}{\sqrt{n}}X\|_1 $ in Inequality (\ref{eq:proof}), if we in addition have that $\|X\|_1 = O(\sqrt{n})$, we have 
\begin{equation}
\begin{split}
\gamma \|\nu\|_2^2 \leq \frac{1}{n} \|X\nu\|^2_2 &\leq \frac{1}{n}\|X^TX\nu\|_\infty\|\nu\|_1\\
&\leq O(\lambda\sqrt{S})\| \nu\|_2.
\end{split}
\end{equation}
Thus we obtain that 
\begin{equation}
\begin{split}
\|\nu\|_2 \leq O(\sqrt{\frac{S\log p}{n}}).
\end{split}
\end{equation}
\end{proof}

\subsection{Proof of Theorem \ref{thm:predictError}}\label{proof:thm:predictError}
\begin{proof}
The proof simply follows from the previous proofs:
\begin{equation}
\begin{split}
\frac{1}{n} \|X\nu\|^2_2 &\leq \frac{1}{n}\|X^TX\nu\|_\infty\|\nu\|_1\\
&\leq O(\lambda\sqrt{S})\| \nu\|_2\\
& = O(\frac{S\log p}{n}).
\end{split}
\end{equation}
\end{proof}



\subsection{Proof of Proposition \ref{prop:maxLambda}}
\label{proof:prop:maxLambda}
\begin{proof}
In order to make the proof more clear, we use the equivalenty statistical form, i.e., Problem (\ref{eq:formulation1}) for the proof.

We will start with the simplest case where $p = 1$.
In this case, we have the objective as
\begin{equation}
\label{eq:formulationLasso1}
\min \|Y - X\beta\|_\infty + \lambda |\beta|,
\end{equation}
where $Y, X \in \mathbb{R}^n$, $\beta \in \mathbb{R}$. 
By the definition of the $\ell_\infty$ norm, we can further write the objective in (\ref{eq:formulationLasso1}) as 
\begin{equation*}
\begin{split}
  &\min \|Y - X\beta\|_\infty + \lambda |\beta|\\
=&\min \max_{i = 1, \cdots, n}\{|Y_i - X_i\beta|\} + \lambda |\beta|.
\end{split}
\end{equation*}
It can be directly seen that this is a convex problem. 
For the first part $\max_{i = 1, \cdots, n}\{|Y_i - X_i\beta|\}$, which is the maximum of a set of convex functions, the resulting function is still convex. 
The second part $\lambda |\beta|$ itself is convex.
Finding the minimum to (\ref{eq:formulationLasso1}) is equivalent to find a stationary solution to it.

On the one hand, for any $\beta_0>0$, we have $\lambda |\beta| = \lambda \beta$. 
Let $i \in I = \{i: |Y_i - X_i\beta_0| = \|Y - X\beta_0\|_\infty\}$, we have in a small neighborhood of $\beta_0$, $\|Y - X\beta\|_\infty = |Y_i - X_i\beta|$.
Thus in this small neighborhood $[\beta_-, \beta_+]$, where $\beta_-, \beta_+ \geq 0$,  the objective is simply
\begin{equation*}
f(\beta) = |Y_i - X_i\beta| + \lambda \beta.
\end{equation*}
By first order condition of the above function $f(\beta)$, we have
$$\partial f(\beta) = \lambda  -  \mbox{sign}(Y_i - X_i\beta)X_i.$$
According to the assumption on $\lambda$, the function $f(\beta)$ is increasing in the small neighborhood $[\beta_-, \beta_+]$.

On the other hand, for any $\beta_0<0$, we have $\lambda |\beta| = -\lambda \beta$. 
Let $i \in I = \{i: |Y_i - X_i\beta_0| = \|Y - X\beta_0\|_\infty\}$, we have in a small neighborhood of $\beta_0$, $\|Y - X\beta\|_\infty = |Y_i - X_i\beta|$.
Thus in this small neighborhood $[\beta_-, \beta_+]$, where $\beta_-, \beta_+ \leq 0$,  the objective is simply
\begin{equation*}
f(\beta) = |Y_i - X_i\beta| - \lambda \beta.
\end{equation*}
By first order condition of the above function $f(\beta)$, we have
\begin{equation*}
\partial f(\beta) = -\lambda  -  \mbox{sign}(Y_i - X_i\beta)X_i.
\end{equation*}
According to the assumption on $\lambda$, the function $f(\beta)$ is decreasing in the small neighborhood $[\beta_-, \beta_+]$.

Using the two observations above, we can obtain that $\beta = 0$ is the unique global minimum for the objective in (\ref{eq:formulationLasso1}).

Now, we will generalize the above proof in the 1D case to high-dimensional case, where we assume $\beta \in \mathbb{R}^p$, $p > 1$.

For any $\beta_0 \in \mathbb{R}^p$, let $J = \{j: \beta_{0j} \neq 0\}$. 
Define the neighborhood of $\beta_0$ as $\mathcal{N} = \{\beta \in \mathbb{R}^p: \beta_j = 0, \mbox{ for } j \not\in J; \beta_j \in [\beta_{0j} - \epsilon, \beta_{0j} + \epsilon], \mbox{ for } j \in J\}$, where $\epsilon > 0$ is chosen such that $(\beta_{0j} - \epsilon)(\beta_{0j} + \epsilon)>0$ for all $j \in J$.

Let $i \in I = \{i: |Y_i - X_i\beta_0| = \|Y - X\beta_0\|_\infty\}$, we have that in  the small neighborhood $\mathcal{N}$ of $\beta_0$, $\|Y - X\beta\|_\infty = |Y_i - X_i\beta|$.
Thus within this small neighborhood $\mathcal{N}$,  the objective is simply
\begin{equation*}
f(\beta) = |Y_i - X_i\beta| + \lambda \|\beta\|_1.
\end{equation*}
By first order condition of the above function $f(\beta)$, we have
$$\partial f(\beta) = \lambda \mbox{sign}(\beta)  -  \mbox{sign}(Y_i - X_i\beta)X_i,$$
where $\mbox{sign}(\beta) \in \mathbb{R}^p$ is the indicator vector of the signs for $\beta$. 
According to the assumption on $\lambda$, $(\partial f(\beta))_k > 0$ for $k \in \mathbb{K}_1 = \{k: \beta_{0k}>0\}$;  $(\partial f(\beta))_k < 0$ for $k \in \mathbb{J}/\mathbb{K}_1$.
Thus we conclude that $\beta = 0$, is the unique global minimum for the objective in (\ref{eq:formulationLasso}), which is $[F_1', F_2']'= 0$ in the Problem (\ref{eq:objforce_l1}).
\end{proof}

\section{Algorithm \ref{alg:admm2}}
 \label{add2}
\begin{algorithm}[H]
\caption{The ADMM for estimating $y_3=(y_2)_S=[(F_1)_{S_1};(F_2)_{S_2}]$}.
\begin{algorithmic}[1]
\STATE Input: $\rho, \lambda, e_3, e_4, n, M, B, U_1, \psi_1, U_2, \psi_2, A, b$
\STATE Initialize $z^0, u^0$
\STATE \emph{loop}:
\FOR {$k=1: K$}
\STATE $u_1^k=\{u_i^k\}_{i=1,...,2n},u_3^k=\{u_i^k\}_{i=1+2n,...,2n+M}$,
\STATE $z_1^k=\{z_i^k\}_{i=1,...,2n},z_3^k=\{z_i^k\}_{i=1+2n,...,2n+M}$,
\STATE $y_1^{k+1}=z_1^k-u_1^k-P_{B_ {\|\cdot\|_1 [0,1]}}(z_1^k-u_1^k)$,
\STATE $y_3^{k+1}=z_3^k-u_3^k$,
\STATE $y^{k+1}=[(y_1^{k+1})', (y_3^{k+1})']'$,
\STATE $z^{k+1}=y^{k+1}+u^k-E'(EE')^{-1}(Ez-b)$,
\STATE $u^{k+1}=u^k+y^{k+1}-z^{k+1}$, 
\STATE $r^{k+1}=y^{k+1}-z^{k+1}$,
\STATE $s^k=\rho(z^{k+1}-z^k)$,
\STATE $e_1=\sqrt{m}e_3+e_4\max\{\|z^{k+1}\|_2,\|y^{k+1}\|_2\}$,
\STATE $e_2=\sqrt{m}e_3+e_4\|\rho u^{k+1}\|_2$,
\IF{$\|r^{k+1}\|_2\leq e_1$ and $\|s^{k+1}\|_2\leq e_2$}
\RETURN $y_3$
\ENDIF
\ENDFOR
\end{algorithmic}
 \label{alg:admm2}
\end{algorithm}
\end{document}